\documentclass[%
 reprint,
superscriptaddress,
 amsmath,amssymb,
 aps,
]{revtex4-2}

\usepackage{graphicx}
\usepackage{dcolumn}
\usepackage{bm}
\usepackage{hyperref}
\usepackage{fullpage}
\usepackage{amsthm}
\usepackage{bbold}
\usepackage{color}
\usepackage{enumitem}

\newcommand{\e}{\mathrm{e}}
\renewcommand{\i}{\mathrm{i}}
\newcommand{\eqdef}{\overset{\text{def.}}{=}}

\renewcommand{\d}{\mathrm{d}}

\newcommand{\sinc}[1]{\mathrm{sinc}\left(#1\right)}

\newcommand{\h}{\hat{H}}

\newcommand{\hint}{\hat{H}_{\text{int}}}
\renewcommand{\a}{\hat{a}}
\newcommand{\ad}{\hat{a}^{\dagger}}
\newcommand{\aw}{\hat{a}_\omega}
\newcommand{\awd}{\hat{a}_\omega^{\dagger}}
\newcommand{\bn}{\hat{b}_n}
\newcommand{\bnd}{\hat{b}_n^{\dagger}}

\newcommand{\hc}{\text{h.c.}}

\newcommand{\U}{\hat{\mathcal{U}}}
\renewcommand{\eqref}[1]{Eq.~(\ref{#1})}
\newcommand{\R}{\mathbb{R}}

\newtheorem{theorem}{Theorem}
\newtheorem{corollary}{Corollary}[theorem]
\newtheorem{lemma}[theorem]{Lemma}
\newtheorem*{remark}{Remark}

\begin{document}

\title{Making Quantum Collision Models Exact}

\author{Thibaut Lacroix}
\email{thibaut.lacroix@uni-ulm.de}
\affiliation{
Institut f\"ur Theoretische Physik \& IQST, Albert-Einstein-Allee 11, Universit\"at Ulm, D-89081 Ulm, Germany
}

\author{Dario Cilluffo}
\email{dario.cillufo@uni-ulm.de}
\affiliation{
Institut f\"ur Theoretische Physik \& IQST, Albert-Einstein-Allee 11, Universit\"at Ulm, D-89081 Ulm, Germany
}

\author{Susana F. Huelga}
\email{susana.huelga@uni-ulm.de}
\affiliation{
Institut f\"ur Theoretische Physik \& IQST, Albert-Einstein-Allee 11, Universit\"at Ulm, D-89081 Ulm, Germany
}
\author{Martin B. Plenio}
\email{martin.plenio@uni-ulm.de}
\affiliation{
Institut f\"ur Theoretische Physik \& IQST, Albert-Einstein-Allee 11, Universit\"at Ulm, D-89081 Ulm, Germany
}

\date{\today}

\begin{abstract}
Quantum collision describe open quantum systems through repeated interactions with a coarse-grained environment.
However, a complete certification of these models is lacking, as no complete error bounds on the simulation of system observables have been established.
Here, we show that Markovian and non-Markovian collision models can be recovered analytically from chain mapping techniques starting from a general microscopic Hamiltonian.
This derivation reveals a previously unidentified source of error -- induced by an unfaithful sampling of the environment -- in dynamics obtained with collision models that can become dominant for small but finite time-steps.
With the complete characterization of this error, all collision models errors are now identified and quantified, which enables the promotion of collision models to the class of numerically exact methods.
To confirm the predictions of our equivalence results, we implemented a non-Markovian collision model of the Spin Boson Model, and identified, as predicted, a regime in which the collision model is fundamentally inaccurate.
\end{abstract}

\maketitle

\section{Introduction}

Quantum collision models offer an intuitive and versatile framework for describing open quantum systems. 
Since their initial formulation in Ref.~\cite{Rau}, these models have become prominent in areas such as weak measurement theory~ \cite{Brun,WisemanMilburnBook}, quantum thermodynamics~\cite{EspositoThermoPRX,Rodrigues_2019}, and quantum optics~\cite{Ciccarello}. 
Applications range from micromaser emission theory~\cite{PhysRevA.34.3077,Filipowicz:86} to waveguide quantum electrodynamics~\cite{Cilluffo_2020}. 
The central idea in collision models is that the system of interest interacts sequentially (``collides") with a set of ancillae representing the environmental degrees of freedom. 
Depending on whether these ancillae are independent and continuously refreshed or correlated and recycled, collision models can capture both Markovian~\cite{Brun,ziman_diluting_2002,scarani_thermalizing_2002, ziman_description_2005, ziman_open_2010, cattaneo_collision_2021} and non-Markovian dynamics~\cite{pellegrini_non-markovian_2009, bodor_structural_2013, CPGCollision, 10.1007/978-3-030-31146-9_3, ciccarello_quantum_2022}. 
However, as with master equations used to model open quantum systems, rigorously benchmarking the accuracy of predictions remains challenging, as it requires rigorously bounding numerical errors as a function of convergence parameters.

In this work, we demonstrate that quantum collision models can enter the domain of numerically exact techniques. 
Specifically, we show that both Markovian and non-Markovian collision models can be derived through chain mapping techniques -- a class of numerically exact methods~\cite{chin_exact_2010, woods_simulating_2015, mascherpa_open_2017}. 
This analytical connection provides a guideline for determining the appropriate coarse-graining timescale for collision models and reveals a unique spectral density sampling error in non-Markovian models, which can exceed other error sources affecting system dynamics. 

Our findings reveal that non-Markovian collision models remain valid and accurate only when the coarse-graining time step between collisions, $\Delta t$, is chosen such that $\Delta t < \frac{\pi}{\omega_c}$ where $\omega_c$ represents the cutoff angular frequency of the bath spectral density. 
This foundational result enhances simulation accuracy in open quantum systems, offering both the collision model and chain mapping communities new insights into the limitations and potential of these techniques.
With all error sources now characterized through equivalence with chain mapping, collision models are elevated to numerically exact methods.

The remainder of the paper is organized as follows: in Sec.~\ref{sec:methods}, we summarize the primary aspects of both methods and establish notation. 
In Sec.~\ref{sec:non-Markovian}, we demonstrate that collision models can recover non-Markovian dynamics using chain mapping. 
Sec.~\ref{sec:markovian} extends this equivalence to the Markovian regime. 
In Sec.~\ref{sec:application}, we apply these results to (i) identify a new error source in non-Markovian collision models and (ii) validate our predictions using the Spin Boson Model (SBM). We conclude with a discussion of the implications of this equivalence in Sec.~\ref{sec:conclusion}.

\section{Overview of the two methods\label{sec:methods}}
We consider, in the Schr\"odinger picture, a general Hamiltonian where a non-specified system interacts linearly with a bosonic environment
\begin{align}
    \h =& \h_S + \int_0^{\infty} \d\omega \, \hbar\omega\ad_\omega\a_\omega \nonumber\\
    & + \hat{A}_S\int_0^{\infty} \d\omega  \sqrt{J(\omega)}\left(\a_\omega + \ad_\omega\right)\label{eq:vibronic}
\end{align}
where $\a_\omega\ (\ad_\omega)$ is a bosonic annihilation (creation) operator for a normal mode of the environment with angular frequency $\omega$, $\hat{A}_S$ is a system operator, and $J(\omega)$ is the bath spectral density (SD) and encodes the coupling strength between the system and the bath modes.
There exist several definition of non-Markovianity~\cite{breuer_colloquium_2016, rivas_quantum_2014,li_concepts_2018, pollock_operational_2018}. 
In this work, we adopt the perspective commonly used in quantum optics, where any spectral density (SD) that is not flat is considered indicative of a non-Markovian environment.

    \subsection{Collision models\label{sec:collision}}

The fundamental concept behind quantum collision models (CMs) is the characterization of the interaction between a quantum system $S$ and its environment (or bath) $E$ as arising from repeated interactions with auxiliary systems, referred to as probes (or ancillae), which collectively represent the environment and share the same initial state $\eta$.
The system evolves through a sequence of pairwise interactions with each probe, which we call \textit{collisions}.
A Markovian CM is defined by the following properties:
\begin{enumerate}[label=C\arabic*]
\item the probes are uncorrelated, e.g.~the initial state of the bath is $(\eta \otimes \eta \otimes...)$;\label{cond:1}
\item probes do not interact with each other;\label{cond:2}
\item each probe is discarded after the interaction with the system and is replaced with a fresh one before the next collision. \label{cond:3}
\end{enumerate}
Additionally, we require that system and environment are uncorrelated at the initial time:
\begin{align}
\sigma_0 = \rho_0 \otimes (\eta \otimes \eta \otimes...)\,,
\end{align}
where subscript $0$ indicates the initial time, $\sigma$ the joint system-environment state and $\rho_0$ is the initial state of $S$. 
The conditions \ref{cond:1}--\ref{cond:3} are fully consistent with the second-order perturbation theory derivation of the Markovian master equation for a discrete dynamics.
Within these assumptions the dynamics of $S$ is decomposable into a sequence of elementary completely-positive maps and thus its temporal evolution can be effectively described through a Master Equation in Lindblad form in the continuous-time limit \cite{ciccarello_quantum_2022, ziman_open_2010}.
When one or more of the aforementioned assumptions is violated, this is no longer possible. This is often interpreted as the introduction of memory effects into the time evolution of the system.
In a general context, describing the dynamics of an open system through collisions necessitates the proper treatment of the Hamiltonian governing the interaction between the system and its surrounding environment. This involves deriving the discretized system-environment coupling Hamiltonian from a microscopic model that accounts for the interactions between the system and the bath.
Starting from the general model in \eqref{eq:vibronic} we can move to the interaction picture with respect to the bath Hamiltonian
\begin{align}
    \h^I(t) &= \h_S + g \hat{A}_S\int_{0}^{\infty} \d\omega \sqrt{J(\omega)}\left(\aw\e^{-\i\omega t} + \awd\e^{\i\omega t}\right)\, ,
    \label{eq:int_pic_0}
\end{align}
where we have scaled the SD with a coupling strength $g \in \mathbb{R}$ for later convenience,
and define the time-domain ladder operators
\begin{align}
    \a(t) &= \frac{1}{\sqrt{2 \pi}}  \int_{0}^{\infty} \d\omega \aw\e^{-\i\omega t} \,.
    \label{eq:time_ops}
\end{align}
It's important to highlight that in what follows we will deliberately avoid moving to the interaction picture with respect to the system's Hamiltonian and refrain from introducing the rotating wave approximation (RWA). 
While we acknowledge that these two approximations play a critical role in establishing a self-consistent definition of Markovian collision models \cite{CPGCollision}, we have chosen to maintain a more general model for the purpose of comprehensive comparison with chain mapping.

In terms of the time-domain operators, the final discrete-time generator of the joint system-bath dynamics reads (see Appendix \ref{sec:collis} for details)
\begin{align}
    \h^I_n = &\h_S + \nonumber\\
    \frac{\hat{A}_S}{\Delta t} &\int_{t_{n-1}}^{t_{n}}\!\!\!\!\!\!\d t \sum_m \int_{t_{m-1}}^{t_{m}}\!\!\!\!\!\!\d t'  g\left(\mathcal{F}[\sqrt{J}](t-t') ~ \hat{a}(t')+ \hc\right) \ ,
    \label{eq:coll_int_intermediate}
\end{align}
with the Fourier transform of the spectral density defined as
\begin{align}
    \mathcal{F}[\sqrt{J}](t-t') &= \frac{1}{\sqrt{2\pi}} \int_{-\infty}^{\infty} \!\!\!\!\!\! \d\omega \sqrt{J(\omega)} \e^{-\i\omega (t-t')} \ .
\end{align}

If we now replace the Fourier transform of the SD with its average over $\Delta t$ we find
\begin{align}
    \h^I_n & \simeq \h_S + \hat{A}_S \sum_{m} (W_{nm} \a_m + \hc) \,,
   \label{eq:c_model}
\end{align}
with 
\begin{align}
W_{nm} &= \frac{1}{\Delta t} \int_{t_{m-1}}^{t_{m}}\!\!\! \d t' \int_{t_{n-1}}^{t_{n}} \d t \, g  \mathcal{F}[\sqrt{J/\Delta t}](t-t')\,\label{eq:kern},
\\
\a_m &= \frac{1}{\sqrt{\Delta t}} \int_{t_{m-1}}^{t_{m}}\!\!\!\!\!\!\d t' ~ \hat{a}(t')\,.\label{eq:ancilla}
\end{align}
The Equations \eqref{eq:c_model}, \eqref{eq:kern} and \eqref{eq:ancilla} collectively define the effective quantum collision model describing our dynamics: the system interacts with a set of time-bin modes defined by the ladder operators $(\a_m, \a^\dag_m)$, which act as the ancillae.
Note that, according to \eqref{eq:c_model}, the system couples nonlocally to all the ancillae with coupling rate $W_{nm}$.
Figure~\ref{fig:schematics-models}~(b) shows a schematic drawing of collision models.
We retrieve the condition \ref{cond:3} if, after performing the RWA \cite{CPGCollision} in the interaction picture with respect to the system's Hamiltonian, we put $\mathcal{F}[\sqrt{J}](s-t') = \delta(s-t')$ that directly implies $W_{nm} = \delta_{nm}$ making the system only interacts with a single ancilla at once. Note that in the frequency space this corresponds to a perfectly flat coupling. Conversely, in the other cases we are describing a system interacting with a colored-noise bosonic reservoir \cite{10.1007/978-3-030-31146-9_3}.

\begin{figure}
    \centering
    \includegraphics[width=0.5\textwidth]{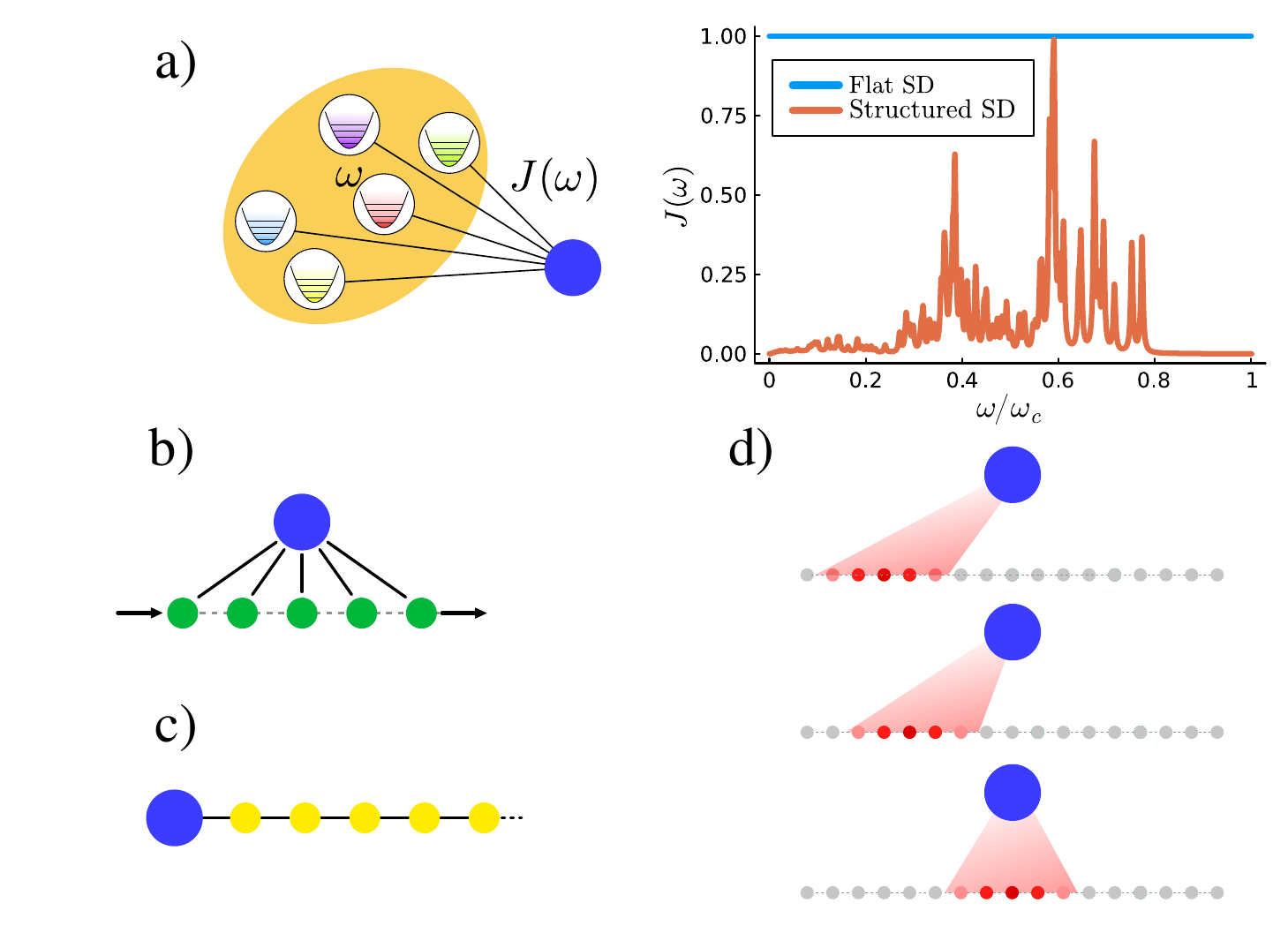}
    \caption{a) A quantum system (blue disk) is interacting with an environment made of a continuum of \emph{non-interacting} bosonic modes of angular frequencies $\omega$. The strength of the interaction between the system and a given mode is encoded in the bath spectral density (SD) $J(\omega)$. Markovian baths are described by a flat (i.e.~constant) SD. A non-flat, i.e.~structured, SD generally induces a non-Markovian dynamics. b) Collision models construct \emph{non-interacting} bosonic temporal-modes on a coarse-grained timescale that experience a finite number of interactions (collisions) with the system before being discarded (refreshed). c) The chain mapping technique maps the bosonic environment into a non-uniform semi-infinite chain of \emph{interacting} bosonic modes such that the system only couples to the first mode of the chain. d) Chain mapping can be reformulated to make the modes \emph{non-interacting} and coupling sequentially to the system for a finite amount of time. This reformulation is equivalent to collision models.}
    \label{fig:schematics-models}
\end{figure}

    \subsection{Chain mapping\label{sec:chainmapping}}

Let us consider the Hamiltonian presented in Eq.~(\ref{eq:vibronic}).
We can introduce a unitary transformation of the continuous normal modes $\aw$ to an infinite discrete set of interacting modes $\bn$~\cite{chin_exact_2010}
\begin{align}
    \aw &= \sum_{n=0}^{\infty} U_n(\omega)\bn = \sum_{n=0}^{\infty} \sqrt{J(\omega)}P_n(\omega)\bn\ , \label{eq:chain-mapping}
\end{align}
where $P_n(\omega)$ are real orthonormal polynomials such that
\begin{align}
    \int_{0}^{\infty} \d\omega \,P_n(\omega)P_m(\omega)J(\omega) = \delta_{n,m}\ ;
\end{align}
and the inverse transformation is
\begin{align}
    \bn &= \int_0^{\infty} \d\omega \, U_n(\omega)\aw\ . \label{eq:inverse-chain-mapping}
\end{align}
Note that the orthonormality of the polynomials ensures the unitarity of the transformation defined in Eq.~(\ref{eq:chain-mapping}). The mapping from a continuous set of modes to a (still infinite) discrete set might seem counter-intuitive, however it is a direct consequence of the separability of the underlying Hilbert space.
Under this transformation, the Hamiltonian in Eq.~(\ref{eq:vibronic}) becomes (see Appendix~\ref{sec:orthopol})
\begin{align}
    \h = \h_S &+ \sum_{n=0}^{\infty}\varepsilon_n\bnd\bn + t_n(\hat{b}_{n+1}^\dagger\bn + \hc)\nonumber\\
    &+ \kappa\hat{A}_S(\hat{b}_0 + \hat{b}_0^\dagger)\ . \label{eq:chain-Hamiltonian}
\end{align}

Hence, this mapping transforms the normal bath Hamiltonian into a tight-binding Hamiltonian with on-site energies $\varepsilon_n$ and hopping energies $t_n$.
Another important consequence of this mapping is that now the system only interacts with the first mode $n=0$ of the chain-mapped environment.
Figure~\ref{fig:schematics-models}(c) shows a schematic drawing of this new topology.
The chain coefficients $\varepsilon_n$, $t_n$, and the coupling $\kappa$ depend solely on the SD (see Appendix~\ref{sec:orthopol}).
This makes chain mapping a tool of choice for describing systems coupled to an environment with highly structured SD (e.g. experimentally measured or calculated \textit{ab initio})~\cite{chin_role_2013, alvertis_non-equilibrium_2019, dunnett_influence_2021, caycedo-soler_exact_2022}.
In this new representation, the Hamiltonian in Eq.~(\ref{eq:chain-Hamiltonian}) has naturally a 1D chain topology.
This makes the representation of the joint \{System + Environment\} wave-function as a Matrix Product State (MPS) very efficient~\cite{orus_practical_2014, paeckel_time-evolution_2019}.
The orthogonal polynomial-based chain mapping and the subsequent representation of the joint wave-function as a MPS (and the operators as Matrix Product Operators) are the building blocks of the Time-Evolving Density operator with Orthonormal Polynomials Algorithm (TEDOPA) one of the state-of-the-art numerically exact method to simulate the dynamics of open quantum systems especially in the non-Markovian, non-perturbative regimes both at zero and finite temperatures~\cite{prior_efficient_2010, woods_simulating_2015, tamascelli_efficient_2019, dunnett_simulating_2021, lacroix_unveiling_2021}.
TEDOPA has been applied, for instance, to transport of electronic excitations in the presence of structured vibrational environment~\cite{prior_efficient_2010}, photonic crystals~\cite{prior_quantum_2013}, non-equilibrium steady states~\cite{dunnett_matrix_2021}, molecular systems~\cite{oviedo-casado_phase-dependent_2015, alvertis_non-equilibrium_2019, le_de_extending_2024}, vibration-induced coherence~\cite{chin_role_2013}, or the calculation of absorption spectra of chromophores~\cite{dunnett_influence_2021, hunter_environmentally_2024, lambertson_computing_2024} and pigment-protein complexes~\cite{caycedo-soler_exact_2022, lorenzoni_systematic_2024}.\\

Here we adopt a slightly different starting point and implement the chain mapping introduced in \eqref{eq:chain-mapping} after moving to the interaction picture with respect to the bath Hamiltonian, the Hamiltonian in \eqref{eq:int_pic_0} reads
\begin{align}
    \h^{I}(t) &= \h_S + \hat{A}_S\sum_{n=0}^{\infty} \left(\gamma_{n}(t)\bn + \gamma_{n}^{*}(t)\bnd\right)\ ,
    \label{eq:simp_ham_chain}
\end{align}
where the $\bn$ operators are the discrete chain modes defined in \eqref{eq:chain-mapping} and the time-dependent coupling coefficients are
\begin{align}
    \gamma_n(t) &= g\int_0^{\infty} \d\omega P_n(\omega)\e^{-\i\omega t}J(\omega) \, .\label{eq:time-dep-couplings}
\end{align}
It can also be noted that the coupling coefficient defined by Eq.~(\ref{eq:time-dep-couplings}) can be expressed as a Fourier transform
\begin{align}
    \gamma_n(t) = g\sqrt{2\pi}\mathcal{F}[P_n J](t)\ ,\label{eq:Fourier-cc}
\end{align}
where $\mathcal{F}[\circ]$ is the Fourier transform of $\circ$.
In this new representation of the system and the environment, the chain modes are now \textit{non-interacting} and all coupled to the system with time-dependent coupling~\cite{liu_improved_2021}.
In the interaction picture the chain mapping brings us from a star topology (see Fig.~\ref{fig:schematics-models}a) of the system-environment interactions with constant coupling strengths $\sqrt{J(\omega)}$ to another star topology where the couplings between the system and the environmental modes are time-dependent $\gamma_n(t)$.

\section{Equivalence in the Non-Markovian case\label{sec:non-Markovian}}

In this section we prove that non-Markovian collision models can be recovered from chain mapping.
\begin{theorem}
    For any positive bath spectral density $J(\omega)$, chain mapping is equivalent to a non-Markovian collision model with $\Delta t = \frac{\pi}{\omega_c}$, where $\omega_c$ is the bath cut-off angular frequency. \label{thm:non-Markovian}
\end{theorem}

In the chain mapping approach there is no fundamental difference between the Markovian and non-Markovian case.
Here we want to discuss the general case of non-Markovian environment, namely when the SD is frequency-dependent. 
The following derivation applies to any SD including, for instance, the highly structured ones found in biological contexts~\cite{caycedo-soler_exact_2022, ratsep_electronphonon_2007}.
As outlined above, the usual chain mapping is to use the unitary transformation defined by the set of orthonormal polynomials with respect to the measure $J(\omega)$ (see Appendix~\ref{sec:orthopol}).
Thus, for different SD the chain operators $\bn$ would be a different linear combination of the normal modes $\aw$.
In any case, the time-dependent coupling coefficients are given by \eqref{eq:Fourier-cc}.
These coupling coefficients have, a priori, an unknown behaviour.

The proof of Thm.~\ref{thm:non-Markovian} relies on noting the following fact.
If we perform the chain mapping unitary transformation in \eqref{eq:chain-mapping} with respect to a flat measure regardless of the nature of the actual SD, we can see that the time-dependent couplings $\gamma_n(t)$ will be given by the convolution of the Fourier transform of the square-root of the SD (i.e.~the frequency-dependent system-environment coupling strength) and the flat measure coupling coefficients $\gamma_n^\text{M}(t)$
\begin{align}
   \gamma_n(t) &= \left(\mathcal{F}[\sqrt{J}]*\gamma_n^{\text{M}}\right)(t)\ .
   \label{eq:non-Markovian_coupling} 
\end{align}

\begin{lemma}
    For a flat SD, the coupling coefficient $\gamma_n^{\text{M}}(t)$ between the system and any chain mode $n$ is non-zero only at a single time $t_n$. \label{lem:gamma}
\end{lemma}

We consider a flat spectral density up to a cut-off frequency $\omega_c$
\begin{align}
    J(\omega) &= \Pi_{\omega_c}(\omega)\ ,
\end{align}
where $\Pi_{\omega_c}(\omega)$ is the indicator function of the interval $[0,\omega_c]$ where it takes the value 1 while vanishing on the complement.
Introducing a frequency cut-off to our environment makes the calculations below more technical, however this is how numerically exact methods such as TEDOPA are implemented in practice.
Hence we believe that the results obtained below will prove more fruitful with the introduction of this frequency cut-off.
With this choice of SD, the orthonormal polynomials defining the chain modes are shifted Legendre polynomials (see Appendix~\ref{sec:orthopol}).
It can be shown that the cut-off frequency $\omega_c$ always corresponds to the cut-off frequency of the bath SD $J(\omega)$ (see Appendix~\ref{app:cut-off-frequency}).

\begin{proof}
The coupling coefficients are given by 
\begin{align}
    \gamma_n^\text{M}(t) &= g\int_0^{\omega_c} \d\omega P_n^\text{shifted}(\omega)\e^{-\i\omega t}\ .\label{eq:markovian-coupling}
\end{align}
The shifted polynomials can be expressed in terms of the regular Legendre polynomials $P_n$ which are defined on the support $[-1, 1]$: $P_n(x) = P_n^\text{shifted}(\frac{x+1}{2})$, with $x = \omega/\omega_c$. Hence, we have
 \begin{align}
     \gamma_n^\text{M}(t) &= g\omega_c\int_0^{1} \d x \,  P_n^\text{shifted}(x)\e^{-\i x \omega_c t}\\
     &= g\omega_c\frac{\e^{-\i\frac{\omega_c t}{2}}}{2} \int_{-1}^{1} \d x \, P_n(x) \e^{-\i x\frac{\omega_c t}{2}}  \label{eq:gamma_regular}\ .
 \end{align}
We can perform a so called plane-wave expansion of the exponential on the Legendre polynomials~\cite{landau_quantum_2013}
\begin{align}
    \e^{-\i x \frac{\omega_c t}{2}} &= 2 \sum_{l=0}^{\infty} \i^l (2l + 1) P_l(x) \sqrt{\frac{\pi}{\omega_c t}} J_{n+\frac{1}{2}}\left(\frac{\omega_c t}{2}\right)\ , 
\end{align}
where $J_{\nu}(\theta)$ is the Bessel function of the first kind.
Inserting this expansion in Eq.~(\ref{eq:gamma_regular}) and using the polynomials orthogonality, we have
\begin{align}
    \gamma_n^\text{M}(t) &= \i^n g\omega_c\e^{-\i\frac{\omega_c t}{2}}\sqrt{\frac{\pi}{\omega_c t}}J_{n+\frac{1}{2}}\left(\frac{\omega_c t}{2}\right)\ . \label{eq:fourier}
\end{align}
We can find the limit of the time-dependent coupling coefficients $\gamma_{n}^\text{M}(t)$ when $\omega_c$ is large by using the asymptotic expansion of the Bessel function $J_{\nu}(\theta)$ for large $\theta$~\cite{olver_asymptotics_1997}
\begin{align}
    \gamma_n^\text{M}(t) &\simeq 2\i^n g \omega_c \e^{-\i\frac{\omega_c t}{2}} \frac{\sin\left(\frac{\omega_c t}{2} - n\frac{\pi}{2}\right)}{\omega_c t}\ ,\label{eq:gamma-sinc}
\end{align}
Taking the limit of infinitely large cut-off frequency (see Appendix~\ref{app:derivation-gamma_n}), we have 
\begin{align}
    \gamma_n^\text{M}(t) &\overset{\omega_c\to\infty}{\cong} 2\pi g \delta\left(t-\frac{n\pi}{\omega_c}\right)\ ,\label{eq:deltan}
\end{align}
Hence, for a flat SD, the coupling coefficient $\gamma_n^\text{M}(t)$ between a chain mode $n$ and the system is non-zero only for $t_n = n\pi/\omega_c = n\Delta t$.
\end{proof}

\begin{remark}
Lemma~\ref{lem:gamma} extends naturally to the exactly Markovian case of a spectral density flat along the whole real line.
In that case the spectral density is chosen to be a rectangular function on the interval $[-\frac{\omega_c}{2}, \frac{\omega_c}{2}]$ to ensure the same bandwidth.
The polynomials are thus directly the Legendre polynomials
\begin{align}
    \gamma_n^\text{M}(t) &= g\int_{-\frac{\omega_c}{2}}^{\frac{\omega_c}{2}} \d\omega \, P_n(\omega)\e^{-\i\omega t}\ ,
\end{align}
from which the same derivation follows leading to the same result.
\end{remark}

Equipped with Lemma~\ref{lem:gamma} we can now prove Thm.~\ref{thm:non-Markovian}
\begin{proof}
The time-dependent coupling coefficients are given by
\begin{align}
   \gamma_n(t) &= \left(\mathcal{F}[\sqrt{J}]*\gamma_n^{\text{M}}\right)(t)= 2\pi g\mathcal{F}\left[\sqrt{J}\right](t - t_n)\ . 
\end{align}
Therefore, the chain-mapped interaction-picture interaction Hamiltonian is
\begin{align}
    \hint^I(t) &= \hat{A}_S\sum_{n=0}^{\infty} \left(2\pi g\mathcal{F}\left[\sqrt{J}\right](t-t_n) \bn + \hc \right)\ .
\end{align}
The time integral of the interaction picture Hamiltonian is the generator of the time-ordered time-evolution operator
\begin{widetext}
\begin{align}
    \int_0^{t = N\Delta t} \d t' \, \hint^I(t') & =\hat{A}_S\sum_{n=0}^{\infty} \left(2\pi\left\{g\int_0^t \d t' \, \mathcal{F}\left[\sqrt{J}\right](t'-t_n) \right\} \bn + \hc \right)\\
    &= \hat{A}_S\sum_{n=0}^{\infty} \left(2\pi\left\{\sum_{m=0}^{N-1}g\int_{m\Delta t}^{(m+1)\Delta t} \d t' \, \mathcal{F}\left[\sqrt{J}\right](t'-t_n)\right\} \bn + \hc \right)\\
    &= \sum_{m=0}^{N-1} \hat{A}_S \left(2\pi\sum_{n=0}^{\infty}\left\{g\int_{m\Delta t}^{(m+1)\Delta t} \d t'\,  \mathcal{F}\left[\sqrt{J}\right](t'-t_n)\right\} \bn + \hc \right)\label{eq:kern-chain}\\
    &= \sum_{m=0}^{N-1} \hat{A}_S \left(\sum_{n=0}^{\infty}W_{mn} \frac{2\pi}{\sqrt{\Delta t}}\bn + \hc \right)\Delta t\ ,
\end{align}
\end{widetext}
where $\Delta t = \frac{\pi}{\omega_c}$ and 
\begin{align}
    W_{mn} \eqdef \frac{g}{\sqrt{\Delta t}}\int_{m\Delta t}^{(m+1)\Delta t} \d t'\,  \mathcal{F}\left[\sqrt{J}\right](t'-t_n)\ .
\end{align}
If we consider $\a_n \eqdef \frac{2\pi}{\sqrt{\Delta t}} \bn$ as an ancilla operator, we recover \eqref{eq:c_model} defining non-Markovian collision models
\begin{align}
    \h^I_n & = \h_S + \hat{A}_S \sum_{m = 0}^{\infty} (W_{nm} \a_m + \hc) \ . \label{eq:non-Markovian-generator}
\end{align}
\end{proof}
We note that, as in collision models (see \eqref{eq:ancilla} and  \eqref{eq:kern}), the ancillae $\a_n$ and collision rates $W_{nm}$  scale as $(\sqrt{\Delta t})^{-1}$.
However, there is a fundamental difference between the collision models rates $W_{nm}$ defined in \eqref{eq:kern} and those obtained from the chain mapping approach in \eqref{eq:kern-chain}.
Indeed, \eqref{eq:non-Markovian-generator} is an exact result: no averaging to decouple a convolution product was performed.
The continuous time limit $\Delta t \to 0$ is widely recognized as a source of challenges in quantum collision models since it demands careful consideration and specialized treatment \cite{ciccarello_quantum_2022}. Remarkably, these challenges do not arise in the context of chain mapping, where the limit $\omega_c\to\infty$ is usually never formally taken.
It is thus interesting to see that these two limits become equivalent within the prescription for the time step $\Delta t = \pi/\omega_c$.
We note that this coarse-grained timescale $\Delta t$ satisfies the Shannon-Nyquist sampling theorem.
For non-vanishing $\Delta t$ the collisional generator in \eqref{eq:non-Markovian-generator} remains valid with collision rates $W_{nm}$ being obtained thanks to \eqref{eq:non-Markovian_coupling} and \eqref{eq:fourier}.
The sequential interaction between the chain modes and the system is preserved by the convolution in Eq.~(\ref{eq:non-Markovian_coupling}).
Yet, depending on the form of $\mathcal{F}\left[\sqrt{J}\right](t)$, several modes can be interacting with the system at a given time, and conversely chain modes interact more than once with the system.
This new representation of the system-bath interaction is represented in Fig.~\ref{fig:schematics-models}(d).
After a certain time, the number $M$ of chain modes a system interacts with can be considered constant.
This is an instance of collision model with multiple non-local collisions \cite{ciccarello_quantum_2022} with $M$ ancillae at a time.

\section{Equivalence in the Markovian case \label{sec:markovian}}
The case of Markovian collision models is a corollary of Thm.~\ref{thm:non-Markovian}.
It follows naturally from Lemma~\ref{lem:gamma} that shows that, for a flat SD, a chain mode $n$ couples to the system only at single time $t_n$.
\begin{corollary}
    If the bath spectral density is flat with a frequency cut-off $\omega_c$ larger than the energy scale of the system (i.e.~a Markovian environment), then chain mapping is equivalent to a collision model with $\Delta t = \frac{\pi}{\omega_c}$.\label{thm:markovian}
\end{corollary}

\begin{proof}
The time-evolution operator in the interaction picture is 
\begin{align}
    \U(t) &= \overset{\leftarrow}{T}\exp\left(-\frac{\i}{\hbar}\int_{0}^{t}\d\tau\,\h^I(\tau)\right)\ \label{eq:time-evolution},
\end{align}
where $\overset{\leftarrow}{T}$ is the time-ordering operation.
Given lemma~\ref{lem:gamma} and Eq.~(\ref{eq:time-evolution}), we have
\begin{align}
    \U(t) &= \overset{\leftarrow}{T}\exp\left(-\frac{\i}{\hbar}\left(\h_S t + \hat{A}_S\sum_{n=0}^{N} \gamma_n \bn + \gamma_n^{*} \bnd \right)\right) \\
    \U(t) &= \overset{\leftarrow}{T}\exp\left(-\frac{\i}{\hbar}\sum_{n=0}^{N}\h^{I}_n\Delta t\right)
\end{align}
where we introduced the coarse-grained timescale $\Delta t = \frac{\pi}{\omega_c}$, $N = t/\Delta t$, $\gamma_n = \int_0^t \gamma_n(\tau)\d\tau = (2\pi)^{\frac{3}{2}} g$. 
All the terms in the sum commute with one another, and we can also assume without loss of generality that they commute with $\h_S$\footnote{This is the same situation as in the derivation of collision models, either $\left[\h_S, \hat{A}_S\right] = 0$, or we move to the interaction picture with respect to the system and bath free Hamiltonians. 
In the `worse' case scenario the evolution operator can be Trotterized.}, thus we have $[\h^I_n , \h^I_m] = 0$.
We can write the time evolution operator as
\begin{align}
    \U(t) &= \U_N\U_{N-1}\ldots\U_{1}\U_{0}\, ,
\end{align}
with $\U_K = \e^{-\frac{\i}{\hbar}\h^{I}_K\Delta t}$.
Hence, we have made explicit that, in the Markovian limit, the time-evolution takes the form of a succession of interactions between the system and individual non-interacting environmental modes, with time-steps $\Delta t$.
\end{proof}

This shows that we recovered a Markovian collision model for bosonic environments starting from the chain mapping of a microscopic Hamiltonian.
Here again, the connection with collision model can be made even more explicit if we recast the interaction part of the argument of the time evolution operator as follows
\begin{align}
    &\int_{0}^{t} \d\tau\, \h_\text{int}^{I}(\tau) = \Delta t \hat{A}_S\sum_{n=0}^{N} \frac{\sqrt{2\pi} g}{\sqrt{\Delta t}}\hat{a}_n + \hc \ ,
    \label{eq:coll_int_mark}
\end{align}
where $\hat{a}_n \eqdef \frac{2\pi}{\sqrt{\Delta t}}\bn$ would play the role of the ancilla operator, and the characteristic factor of $(\sqrt{\Delta t})^{-1}$ of the collision model coupling strength is recovered~\cite{ciccarello_quantum_2022}.\\
If we compare \eqref{eq:coll_int_mark} with \eqref{eq:simp_ham_chain} we can observe that collision models and chain mapping are two different ways to take into account the same time-dependent behavior of the Hamiltonian, which arises when moving to the interaction picture. In collision models the interaction Hamiltonian is fixed in time and the time dependence is represented by the sequential interaction with the time modes whereas in the chain-mapping picture the time dependence is entirely attributed to the coupling $\gamma_n(t)$.

\section{Applications\label{sec:application}}

\subsection{Sources of Error in Collision Models\label{sec:errors}}

From their canonical derivation collision models rely on an expansion of the time-evolution operator to second-order in $\Delta t$ which thus leads to a so called `truncation error' of the reduced system's dynamics of order $\mathcal{O}(\Delta t^3)$~\cite{grimmer_open_2016, ciccarello_quantum_2022}.
In numerical simulations the time-evolution operator is usually approximated using a Trotter-Suzuki decomposition~\cite{suzuki_generalized_1976}, inducing a `Trotter error' that can be matched with the usual truncation error $\mathcal{O}(\Delta t^3)$ by using a second order Troterrization.
The error originating from the truncation of the infinite-dimensional local Hilbert spaces of the bath modes vanishes with the increase of the aforementioned local dimensions~\cite{woods_simulating_2015}.
When combined with tensor networks, another common numerical error is the Singular Value Decomposition truncation error.
Properly choosing the threshold for discarding singular values enables to keep this error lower than the previous ones.

However, for non-Markovian collision models, there is an additional source of error to take into account that also stems from the very derivation of the method: the bath correlation function sampling error.
This sampling error is introduced in \eqref{eq:c_model}, \eqref{eq:kern} and \eqref{eq:ancilla} when averaging the Fourier transform of the square-root of the SD to get rid of the convolution product.
The order of the sampling error of the bath correlation function is a priori unknown and needs to be quantified in order to be compared to the other sources of error.
Given that the SD is non-negative, sampling $\sqrt{J(\omega)}$ gives the same information as sampling $J(\omega)$.
The Shannon-Nyquist sampling theorem tells us that when we sample with a frequency $1/\Delta t$, we can reconstruct the SD up to $\omega = \pi/\Delta t$ using so called `perfect reconstruction' with, for instance, Whittaker's interpolation~\cite{shannon_communication_1949, whittaker_interpolatory_1935, bremaud_mathematical_2002}.
Hence when $\Delta t \leq \pi/\omega_c$ the SD is perfectly sampled, and when $\Delta t > \pi/\omega_c$ a sampling error is introduced.
For Markovian collision model this sampling error does not exist as any time-step $\Delta t$ yields to the exact SD.
That is why a single ancilla is sufficient to describe the dynamics.
However, for non-flat SD this sampling error can become larger than the truncation (or Trotter) error for $\Delta t > \pi/\omega_c$ even though the time step can be made arbitrary small numerically.\\
For the Spin Boson Model (SBM), the impact of this sampling error on the expectation value of an observable can be upper bounded~\cite{mascherpa_open_2017}.
The sampling error on the expectation value $\langle\sigma_z\rangle(t)$ after a single time step $\Delta t$ is 
\begin{align}
    \epsilon_\text{samp} &\leq \exp\left( 4\int_0^{\Delta t}\d t' \int_0^{t'}\d t'' |\Delta C(t' - t'')| \right) - 1\ .
\end{align}
Let us consider an Ohmic SD~${J(\omega) = 2\alpha\omega \Pi_{\omega_c}(\omega)}$,
\begin{align}
    \Delta C(\tau) &= \int_{\frac{\pi}{\Delta t}}^{\omega_c} 2\alpha\omega\e^{-\i\omega \tau}\d\omega\label{eq:diffC}\\
    &= \frac{2\alpha}{\tau^2}\left(\e^{-\i\omega_c\tau}(1 + \i\omega_c \tau) - \e^{-\i\frac{\pi\tau}{\Delta t}}\left(1 + \i\frac{\pi\tau}{\Delta t}\right)\right)
\end{align}
is the difference between the exact bath correlation function and the sampled one.
The sampling error vanishes for $\Delta t \leq \pi/\omega_c$ because the upper and lower integration bounds in \eqref{eq:diffC} are equal.
For ${\Delta t \geq \frac{\pi}{\omega_c}}$ the error is
\begin{align}
    \epsilon_\text{samp} &\leq 2\pi^2\alpha\left(\left(\frac{\omega_c\Delta t}{\pi}\right)^2 - 1\right)\ .
\end{align}
Thus, for a SBM with an Ohmic SD, when $\Delta t \leq \pi/\omega_c$ the leading error is the truncation/Trotter error $\mathcal{O}(\Delta t^3)$, and when $\Delta t > \pi/\omega_c$ the leading error is the sampling error $\mathcal{O}(\Delta t^2)$.

\subsection{Spin Boson Model}

\begin{figure*}[t]
    \centering
    \includegraphics[width=\textwidth]{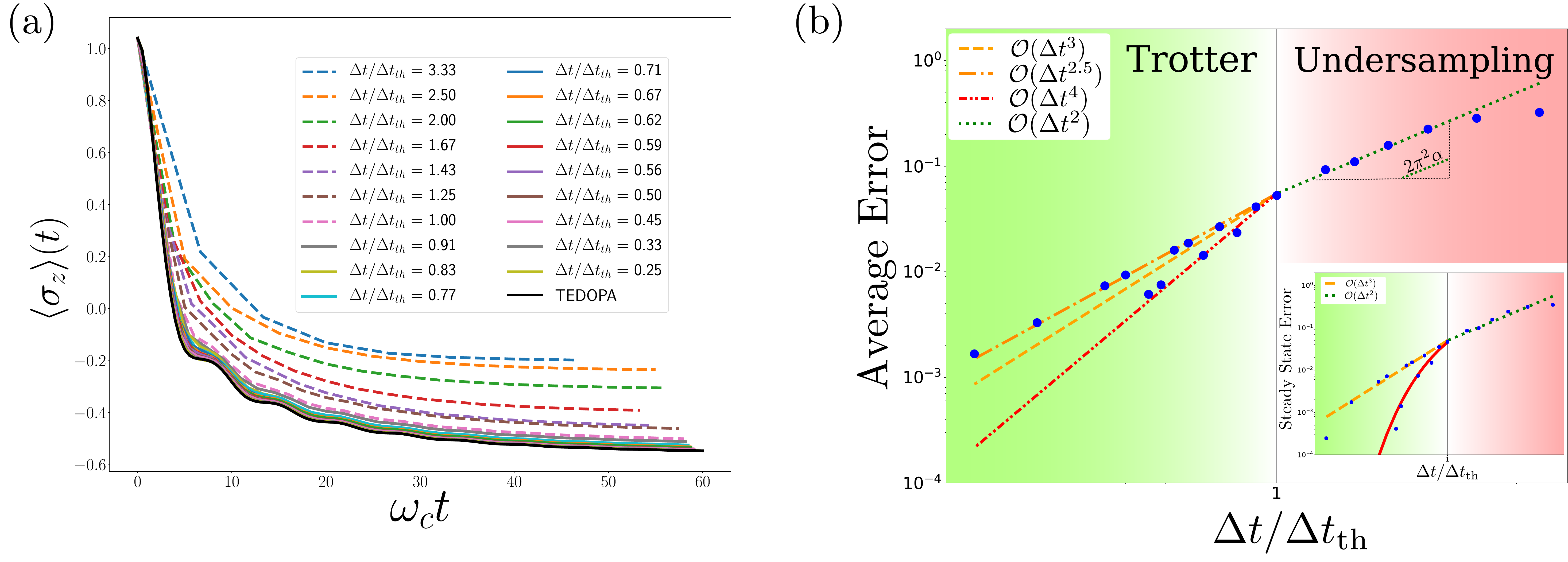}
    \caption{(a) Comparison $\langle\sigma_z\rangle(t)$ between a non-Markovian collision model, for different time steps $\Delta t$, with reference TEDOPA results (black solid line). (b) Main panel: Average error between the collision model dynamics obtained for a given time step $\Delta t$ and the reference results. Inset: Distance between the steady state expectation $\langle\sigma_z\rangle(t\to\infty)$ to the reference results as a function of the collision model time step $\Delta t$, the red solid line is a guide to the eye. We can see that $\Delta t_\text{th}$ is a threshold value separating two distinct scaling regimes: for $\Delta t < \Delta t_\text{th}$ the average and steady state errors scale as $\mathcal{O}(\Delta t^3)$, and for $\Delta t \geq \Delta t_\text{th}$ they scale as $\mathcal{O}(\Delta t^2)$. The simulations parameters are $\omega_0 = 0.2\omega_c$, $\delta = 0$, $\alpha = 0.1$.}
    \label{fig:parabolic}
\end{figure*}

The SBM is a paradigmatic model in the field of OQS.
While being simple -- the model consists of a single spin linearly coupled to a bosonic bath -- its physics is rich (and exhibits non-Markovian behaviour) and it has been used to model magnetic impurities, charge transfer, chemical reactions, strangeness oscillations of the $K^0$ mesons, or decoherence~\cite{breuer_theory_2007, weiss_quantum_2012}.
On top of its dynamics being non-trivial, the model is also non-solvable analytically and has become a test-bed for numerical methods describing open systems.
From Eq.~(\ref{eq:vibronic}) the SBM Hamiltonian is obtained by setting 
\begin{align}
    \h_S &= \frac{\omega_0}{2}\hat{\sigma}_z + \delta\hat{\sigma}_x\ \ \ \text{and}\ \ \ \hat{A}_S = \hat{\sigma}_x\ .
\end{align}
We note that in this model no rotating wave approximation has been performed.
In the following we consider an Ohmic SD with a hard cut-off $J(\omega) = 2\alpha\omega \Pi_{\omega_c}(\omega)$ with $\Pi_{\omega_c}(\omega)$ the rectangular function on $[0, \omega_c]$.
Figure~\ref{fig:parabolic}~(a) shows the expectation value of $\langle\sigma_z\rangle(t)$ obtained with a non-Markovian collision model for several values of $\Delta t$, compared with the dynamics obtained with the regular Schr\"odinger picture chain mapping (i.e. the TEDOPA method) taken as a reference result.
The TEDOPA results are obtained considering 16 environmental modes, and the maximum bond dimension reached during the simulation is $D = 15$.
The non-Markovian collision model has been implemented with tensor networks methods: The \{System + Ancillae\} density matrix is represented as a purified Matrix Product State~\cite{feiguin_finite-temperature_2005, cuevas_purifications_2013} and the time-evolution is performed with the standard time-evolving block decimation (TEBD) method~\cite{vidal_efficient_2004, paeckel_time-evolution_2019}.
The results are obtained with a number of ancillae inversely proportional to $\Delta t$ (e.g. 35 for $\Delta t = 1/\omega_c$, 70 for $\Delta t = 2/\omega_c$, and 280 for $\Delta t = 1/2\omega_c$), and a maximal bond dimension of $D = 32$.
We would like to point out that, to the best of our knowledge, this is the first time that the SBM has been simulated with a collision model.
It has to be noted that, because the cut-off frequency $\omega_c$ of the SD remains `small' in numerical simulations, the threshold time-step in these simulations is $\Delta t_\text{th} = 2/\omega_c$ instead of $\pi/\omega_c$ (see Appendix~\ref{app:numerical-threshold}).
This is due to the asymptotic behaviour of spherical Bessel functions.
On Fig.~\ref{fig:parabolic}~(a) we can see that both the steady state and the transient dynamics are better described when $\Delta t$ diminishes.
For instance, the oscillatory dynamics start to be well caught around $\Delta t = 2/\omega_c$.
The dynamics converges monotonically from above with decreasing time steps.
Figure~\ref{fig:parabolic}~(b) (main panel) shows the average error during the dynamics of the collision model simulations with respect to the reference one.
We can clearly see that there are two different scaling regime separated by the threshold value $\Delta t_\text{th}$.
For time step smaller than the threshold $\Delta t < \Delta t_\text{th}$ we are in a regime where the deviation is dominated by an error $\mathcal{O}(\Delta t^{2.5})$ associated with the second order Trotterization performed to obtain the collision model time-evolution operator.
We also note that for specific values of $\Delta t$ in this regime the error can be smaller than the Trotter error -- which is perfectly legitimate considering that the Trotter scaling is an upper bound.
This might originate from `local' error cancellation.
The investigation of this `super-performance' is beyond the scope of this paper. 
When the time step is larger than the threshold $\Delta t > \Delta t_\text{th}$ we can see a sudden change in the scaling of the error that is now $\mathcal{O}(\Delta t^2)$ (for large $\Delta t$ the error saturates because $\langle\sigma_z\rangle(t)$ decays exponentially to 0).
We attribute this additional source of error to a fundamental inaccuracy of the collision model in this regime, as can be inferred from the equivalence theorem.
When $\Delta t > \Delta t_\text{th}$ the scaling of the errors in our simulations have a slope of $2\pi^2\alpha$ in agreement with the one expected from the discussion in Sec.~\ref{sec:errors}, and thus shows that in the fundamental inaccuracy regime an aliased sampling of the bath correlation function results in an error of order $\mathcal{O}(\Delta t^2)$.
The distance between the steady state expectation value $\langle\sigma_z\rangle(t\to\infty)$ and the reference one for different values of the time step $\Delta t$ is presented in the inset of Fig.~\ref{fig:parabolic}~(b).
Here again we find the same transition between two scaling regimes of the error at $\Delta t_\text{th}$.
For time steps larger than the threshold $\Delta t > \Delta t_\text{th}$ we have a scaling of $\mathcal{O}(\Delta t^2)$ worse than the Trotter one $\mathcal{O}(\Delta t^3)$ which is recovered for time steps smaller than the threshold value $\Delta t < \Delta t_\text{th}$.
These results show that, in order to give physically accurate results, the chosen time step of the collision model has to be lower or equal to the threshold value $\Delta t_\text{th}$.
This prescription gives a consistent definition to how small the time step needs to be to ensure the validity of collision models.

\section{Discussion\label{sec:conclusion}}

In this paper we introduced an analytical derivation of (Markovian and non-Markovian) collision models based on the chain mapping of the environment that places both on the same footing~\footnote{We would like to note that in this work chain mapping has been used as an analytical tool to derive exact collision models.
In a different spirit, collision models and chain mapping have already been used together in the periodically refreshed baths (PReB) approach~\cite{purkayastha_periodically_2021, purkayastha_periodically_2022}. PReB can be understood as a collision model where a system collides with macroscopic environments -- that can be described as chain-mapped environments -- with large time-steps $\tau$ greater than the memory time of the environments.}.
One consequence of this is a prescription for the time step used in collision models that eliminates the environmental sampling error.
This prediction was tested within the paradigmatic Spin Boson Model where we have shown that the predicted time step identifies a threshold value between a regime where the Trotter error dominates and a fundamental inaccuracy regime related to an under-sampling of the bath SD.
The first consequence of this equivalence is to shed light on a previously overlooked source of error in non-Markovian collision models that is larger than the well-known truncation error of collision models.
Taking into account and characterizing this new error enables the promotion of collision models to the class of numerically exact methods, as they otherwise share the good analytical and numerical properties of chain mapping and its associated numerical methods.

Chain mapping techniques can be enriched from this equivalence result.
On the conceptual side, it improves the understanding of the nature of the chain modes that did not have a firmly grounded physical interpretation~\cite{tamascelli_excitation_2020}.
Indeed, chain modes can now be interpreted as temporal modes.
Collision models have been successfully connected to other open quantum system approaches such as stochastic trajectories or input-output formalism, and have become a framework of choice in quantum thermodynamics.
Approaches based on chain mapping could learn from these connections.
The TEDOPA method suffers from the linear growth of the number of chain modes that need to be considered for an increasing simulation time.
Because ancillae that are no longer interacting can be traced out, collision models do not suffer from this limitation.
Recently, it has been shown that connecting a collection of sinks to the truncated chain-mapped environment can circumvent this fundamental limitation at the price of describing the joint \{System + Environment\} state as a density matrix~\cite{nuseler_fingerprint_2022}.
One could ask whether this approach is formally equivalent to the discarding of ancillae in collision models.\\

Even though collision models can be defined from microscopic models they are often stated as a starting assumption.
The equivalence results presented in this paper allow a more systematic derivation of collision models from microscopic models.
Indeed chain mapping can be used to derive a collision model especially in contexts where such a derivation is highly non-trivial (for instance quantum optical systems with non-linear bath dispersion relations~\cite{prior_quantum_2013}).
On the side of implementations, chain mapping can be combined with Matrix Product States to give the TEDOPA method. 
Additionally, we have employed collision models with tensor networks to simulate the dynamics of open systems in a regime far outside regimes where typical approximations (in particular RWA and weak coupling) hold.
This is especially important given that chain mapping is well-defined for any positive spectral density.
This implies that experimentally measured or calculated from first principle methods SDs are also accessible to collision models.
Another important consequence for collision models is related to their extension to fermionic environment, which is currently still an open problem.
However, the formalism of chain mapping for fermionic environments already exists~\cite{nuseler_efficient_2020, kohn_efficient_2021, ferracin_spectral_2024}.
Therefore future work will be devoted to the investigation of fermionic collision models.\\

\begin{acknowledgments}
This work is supported by the ERC Synergy grant HyperQ (grant no 856432) and the BMBF project PhoQuant (grant no 13N16110).
The collision model simulations made use of the \texttt{mpnum} Python library~\cite{suess_mpnum_2017}.
The TEDOPA simulations were performed using the open-source \texttt{MPSDynamics.jl} package~\cite{mpsdynamics_zenodo, mpsdynamicsjl_2024}.
\end{acknowledgments}

%

\newpage

\onecolumngrid

\appendix

\section{Collision model derivation \label{sec:collis}}

In the derivation of the microscopic joint system-environment Hamiltonian \eqref{eq:coll_int_intermediate} 
we made the assumption that the system's characteristic frequencies are centered around a positive value we refer to as $\tilde{\Omega}$ and are confined to a limited bandwidth.
Additionally, we extend the domain of the spectral density to include negative frequencies by setting $J(\omega) = 0$ for $\omega < 0$ and assume that no bath modes with negative frequencies are populated at time $t=0$. This allows us to extend the integration limits in \eqref{eq:time_ops} and consequently in \eqref{eq:int_pic_0}, to encompass the entire real axis

\begin{align}
    \h^I(t) &= \h_S + g \hat{A}_S\int_\R \d\omega  \left(\sqrt{J(\omega)} \int_\R \frac{\d t'}{\sqrt{2\pi}} \a(t')\e^{-\i\omega (t-t')} + \hc \right) \notag
    \\&=
    \h_S + g \hat{A}_S \int_\R \d\omega \int_\R \frac{\d t''}{\sqrt{2\pi}} \left(\mathcal{F}[\sqrt{J}](t'') \e^{i\omega t''}\int_\R \frac{\d t'}{\sqrt{2\pi}}  \a(t')\e^{-\i\omega (t-t')} + \hc \right) \notag
    \\&=
    \h_S + g \hat{A}_S  \int_\R \frac{\d t'}{\sqrt{2\pi}} \int_\R \frac{\d t''}{\sqrt{2\pi}}  \left( \mathcal{F}[\sqrt{J}](t'')  \a(t')
    2\pi\delta\Big((t-t')-t''\Big) + \hc \right) \notag
    \\\h^I(t) &=
    \h_S + g \hat{A}_S   \int_\R \d t' \left( \mathcal{F}[\sqrt{J}](t-t')  \a(t')
     + \hc \right)
    \,,
    \label{eq:int_pic_extended}
\end{align}

with the Fourier transform of the spectral density defined as
\begin{align}
    \mathcal{F}[\sqrt{J}](t-t') &= \frac{1}{\sqrt{2\pi}} \int_{-\infty}^{\infty} \!\!\!\!\!\! \d\omega \sqrt{J(\omega)} \e^{-\i\omega (t-t')} \ .\label{eq:appFourier}
\end{align}

We are able to express the \eqref{eq:int_pic_0} in time domain and to discretize it in units of $\Delta t$, which for now is only assumed small with respect to the inverse of the characteristic frequencies of the system-bath interaction. The microscopical discrete-time evolution generator reads 
\begin{align}
    \h^I_n &= \h_S + \frac{g}{\Delta t}\hat{A}_S \int_{t_{n-1}}^{t_{n}}\!\!\!\!\!\!\d t \int_\R\!\!\d t'  \left(\mathcal{F}[\sqrt{J}](t-t') ~ \hat{a}(t')+ \hc\right) \ ,\label{eq:app-nonM-generator}
\end{align}
which is turned into \eqref{eq:coll_int_intermediate}
by using the same coarse-graining time-scale to split the inner integral, as presented in the main text.
The usual Markovian collision model can be retrieved from \eqref{eq:app-nonM-generator} by extending the flat SD to the whole real axis  thanks to the usual assumptions of weak coupling and separations of time-scales~\cite{ciccarello_quantum_2022, Gross_2018}.\\

Note that with the definition of the Fourier transform given in \eqref{eq:appFourier} we have the following relations
\begin{align}
    \mathcal{F}[1] &= \sqrt{2\pi}\delta\ ,\\
    \text{and\ }\ \mathcal{F}[\delta] &= \frac{1}{\sqrt{2\pi}}\ .
\end{align}

\section{Orthonormal polynomials \label{sec:orthopol}}
    \subsection{Orthogonality, recurrence relation and bath chain mapping}
Let $P_n(\omega)$ be a real polynomial of order $n$
\begin{align}
    P_n(\omega) & = \sum_{k=0}^{n} a_k \omega^k\ ,
\end{align}
where $a_k$ are real coefficients.
Two polynomials are said to be orthonormal with respect to a measure $\d J(\omega) = J(\omega)\d\omega$ if
\begin{equation}
    \int_0^{\infty} P_n(\omega)P_m(\omega) J(\omega)\d\omega = \delta_{n,m}\ .
    \label{eq:orthogonality}
\end{equation}
This orthogonality relation defines a unique family of polynomials (up to multiplication by a real constant).

A useful property of these polynomials is that they obey a recurrence relation
\begin{align}
    P_n(\omega) &= (C_{n-1}\omega - A_{n-1})P_{n-1}(\omega) + B_{n-1}P_{n-2}(\omega)\ ,
    \label{eq:recurrence}
\end{align}
where $A_n$ is related to the first moment of $P_n$, $B_n$ and $C_n$ to the norms of $P_n$ and $P_{n-1}$~\cite{appel_mathematics_2007}.
This recurrence relation can be used to construct the polynomials with the conditions that $P_0(\omega)~=~||p_0||^{-1}~=~\left(\int_{\mathbb{R}^{+}} J(\omega)\d\omega \right)^{-\frac{1}{2}}$ and $P_{-1}(\omega) = 0$, with $||\bullet|| = \left(\int_{\mathbb{R}^{+}} |\bullet|^2 J(\omega)\d\omega \right)^{-\frac{1}{2}}$ the norm of $\bullet$ with respect to the measure $J(\omega)$, and $P_n(\omega) = p_n(\omega)||p_n||^{-1}$ ; where the polynomials $\{p_n\}_{n\in\mathbb{N}}$ are the so called \emph{monic polynomials} where the factor $a_n$ in front of $\omega^{n}$ is equal to 1.\\

If we apply the unitary transformation $U_n(\omega) = \sqrt{J(\omega)}P_n(\omega)$ to the interaction Hamiltonian
\begin{align}
\hat{H}_\text{int} &= \hat{A}_S\int_0^{\infty}\sqrt{J(\omega)}\left(\a_\omega + \ad_\omega\right)\d\omega\\
&= \hat{A}_S \int_0^{\infty} \sqrt{J(\omega)} \sum_n U_n(\omega)(\hat{b}_n + \hat{b}_n^\dagger) \d \omega\\
&= \sum_n\hat{A}_S \int_0^{\infty} J(\omega)P_n(\omega)(\hat{b}_n + \hat{b}_n^\dagger)\d\omega\\
&= \sum_n\hat{A}_S \underbrace{\left(\int_0^{\infty} J(\omega)P_n(\omega)P_{0}(\omega)\d\omega\right)}_{\delta_{n,0}}||p_0||(\hat{b}_n + \hat{b}_n^\dagger)\\
\hat{H}_\text{int} &= ||p_0|| \hat{A}_S (\hat{b}_0 +\hat{b}_0^\dagger)
\label{eq:short-range}
\end{align}
we obtain a new expression where the system couples \emph{only} to the first mode with the coupling strength $||p_0|| \eqdef \kappa$.\\
The same transformation applied to the bath Hamiltonian yields, thanks to the recurrence relation, to the following nearest neighbours hopping Hamiltonian where $\varepsilon_n = A_n C_n^{-1}$ is the energy of the chain mode $n$ and $t_n=C_n^{-1}$ is the coupling between mode $n$ and $n+1$
\begin{align}
    \hat{H}_B &= \int_0^{\infty}  \omega\ad_\omega\a_\omega\d\omega\\
    &= \int_0^{\infty}  \omega \sum_{n, m} J(\omega) P_m(\omega)P_n(\omega) \hat{b}^\dagger_m\hat{b}_n\d\omega\\
    &= \sum_{n, m} \int_0^{\infty} \left(\frac{1}{C_m} P_{m+1}(\omega) + \frac{A_m}{C_m}P_m(\omega) - \frac{B_m}{C_m}P_{m-1}(\omega)\right)P_n(\omega)J(\omega)\d\omega  \hat{b}^\dagger_m\hat{b}_n \\
    \h_B &= \sum_n \varepsilon_n \hat{b}_n^\dagger\hat{b}_n + t_n (\hat{b}_n^\dagger\hat{b}_{n+1} + \hat{b}_{n+1}^\dagger\hat{b}_{n}) \ ,
    \label{eq:tightbinding}
\end{align}
where we used the fact that $-B_{n+1}C_{n+1}^{-1} = C_{n}^{-1}$~\cite{chin_exact_2010}. 
From the new bath and interaction Hamiltonians of Eqs.~(\ref{eq:short-range}) and (\ref{eq:tightbinding}) we can see that the unitary transformation $U_n(\omega)$ transforms the bosonic environment composed of a continuum of independent modes --- the star environment --- into a semi-infinite chain of interacting modes (see Fig.~\ref{fig:schematics-models}(c)).
The chain coefficients $\varepsilon_n$, $t_n$ and $\kappa$ can sometimes by calculated analytically, for instance at zero-temperature, otherwise they can be computed numerically with stable and convergent algorithms~\cite{gautschi_algorithm_1994,PolyChaos}.

    \subsection{Polynomials for a flat SD}
For a flat SD with a hard cut-off at the frequency $\omega_c$, the defining orthogonality relation of the polynomials (Eq.~(\ref{eq:orthogonality})) becomes
\begin{align}
    \int_0^{\omega_c} P_n(\omega)P_m(\omega)\d\omega = \omega_c \delta_{n,m}\ ,
\end{align}
where the system-environment coupling strength $g$ has been absorbed in the system operator $\hat{A}_S\to g\hat{A}_S$ for convenience.
Scaling the frequency with the cut-off frequency $\omega = x\omega_c$, we have
\begin{align}
    \int_{0}^{1}P_n(x)P_m(x)\d x = \delta_{n,m}
\end{align}
where the measure is $J(x) = 1$.
This measure and support define the shifted Legendre polynomials whose analytical expression is
\begin{align}
    P_n^\text{shifted}(x) &= \frac{1}{n!}\frac{\d^n}{\d x^n}(x^2 - x)^n\ .
\end{align}
The first polynomial is $P_0(x) = 1$, and the recurrence coefficients are $A_n = \sqrt{(2n+1)(2n+3)}(n+1)^{-1}$, $B_n = -n(n+1)^{-1}(2n+3)^{\frac{1}{2}}(2n-1)^{-\frac{1}{2}}$, and $C_n = 2\sqrt{(2n+1)(2n+3)}(n+1)^{-1}$.
Hence, the chain coefficients are $\varepsilon_n =\frac{\omega_c}{2}$, $t_n = \frac{\omega_c}{2}(n+1)[(2n + 1)(2n+3)]^{-\frac{1}{2}}$, and $\kappa = \sqrt{2}\omega_c$\ .
We note that these recurrence/chain coefficients can also be recovered from the Ohmic spectral density (and its associated shifted Jacobi polynomials) by setting $s = 0$~\cite{chin_exact_2010, abramowitz_handbook_1964}.
Expressed in terms of angular frequency, the polynomials are
\begin{align}
    P_n(\omega) &= \frac{1}{n!}\frac{\d^n}{\d \omega^n}\left(\frac{\omega^2}{\omega_c} - \omega\right)^n\ .
\end{align}

\section{Cut-off frequency of $\gamma_n^\text{M}(t)$ \label{app:cut-off-frequency}}
A realistic bath spectral density $J(\omega)$ will display a maximum frequency $\omega_c$ such that $J(\omega > \omega_c) = 0$.
In complete generality, the SD can thus be written as
\begin{align}
    J(\omega) = \mathcal{J}(\omega)\Pi_{\omega_c}(\omega)\ ,
\end{align}
where $\Pi_{\omega_c}(\omega)$ is the indicator function of the interval $[0,\omega_c]$ where it takes the value 1 while vanishing on the complement, and $\mathcal{J}(\omega)$ is the extension of the SD to the whole real line.
As explained in Sec.~\ref{sec:non-Markovian}, when performing the chain-mapping with respect to a flat SD the time-dependent coupling coefficients become (\eqref{eq:time-dep-couplings})
\begin{align}
        \gamma_n(t) &= g\int_0^{\infty} \d\omega P_n(\omega)\e^{-\i\omega t}\sqrt{J(\omega)}\\
        &= g\int_0^{\infty} \d\omega P_n(\omega)\e^{-\i\omega t}\sqrt{\mathcal{J}(\omega)}\Pi_{\omega_c}(\omega)\\
        \gamma_n(t) &= \left(\mathcal{F}[\sqrt{\mathcal{J}}]*\gamma_n^{\text{M}}\right)(t)\ ,
\end{align}
where $\gamma_n^\text{M}(t)$ corresponds to \eqref{eq:markovian-coupling} and $\mathcal{F}[\circ]$ is the Fourier transform of $\circ$.
Alternatively, it is also possible to define $\gamma_n^\text{M}(t)$ with an arbitrary cut-off frequency $\Omega$, in such a case it can be be shown that the orthogonality of the Legendre polynomials insures that $\Omega = \omega_c$ and that \eqref{eq:fourier} is recovered.

\section{Derivation of $\gamma_n(t)$ for a Markovian environment \label{app:derivation-gamma_n}}
\subsection{Asymptotic limit of $\gamma_n(t)$}

The asymptotic expansion of the spherical Bessel function $\sqrt{\frac{\pi}{\omega_c t}}J_{n+\frac{1}{2}}(\frac{\omega_c t}{2})$ for large $\omega_c$ is~\cite{olver_asymptotics_1997}

\begin{align}
    \sqrt{\frac{\pi}{\frac{\omega_c t}{2}}}J_{n+\frac{1}{2}}\left(\frac{\omega_c t}{2}\right) &\overset{\omega_c\to\infty}{=} \left(\frac{2}{\omega_c t} + \mathcal{O}\left(\left(\frac{\omega_c t}{2}\right)^{-3}\right) \right) \sin \left(\frac{\omega_c t}{2} - \frac{\pi  n}{2}\right)\nonumber\\
    &\ \ \ \ \ \ \ \ \ \ \ +\left(\frac{n (n+1)}{2 (\frac{\omega_c t}{2})^2} + \mathcal{O}\left(\left(\frac{\omega_c t}{2}\right)^{-4}\right)\right) \cos \left(\frac{\omega_c t}{2} - \frac{\pi  n}{2}\right)\\
     &\overset{\omega_c\to\infty}{=} \frac{2}{\omega_c t}\sin \left(\frac{\omega_c t}{2} - \frac{\pi  n}{2}\right) + \mathcal{O}\left(\left(\omega_c t\right)^{-3}\right)\ .
\end{align}

We start by recalling the definition of the Dirac delta in terms of a limit of $\mathrm{sinc}$ function
\begin{align}
    \delta(x) = \lim_{\ell\to 0}\frac{1}{\ell}\sinc{\pi\frac{x}{\ell}}\ ,
\end{align}
where we use the so-called `physicist' convention $\mathrm{sinc}(x) = \sin(x)/x$.\\

Proceeding by identification with \eqref{eq:gamma-sinc} with $1/\ell = \omega_c$
\begin{align}
    \delta\left(\frac{t}{2\pi} - \frac{n}{2\omega_c}\right) &= \lim_{\omega_c \to \infty} 2\frac{\sin(\frac{\omega_c t}{2} - n \frac{\pi}{2})}{t - n \frac{\pi}{\omega_c}} = \lim_{\omega_c \to \infty} 2\frac{\sin(\frac{\omega_c t}{2} - n \frac{\pi}{2})}{t}  = 2\pi \delta\left(t - n \frac{\pi}{\omega_c}\right)\ .
\end{align}

Thus, we have for the flat spectral density coupling strength in \eqref{eq:gamma-sinc}
\begin{align}
    \gamma_n^\text{M}(t) &=\lim_{\omega_c\to\infty} \i^n g \e^{-\i\frac{\omega_c t}{2}} 2\frac{\sin\left(\frac{\omega_c t}{2} - n\frac{\pi}{2}\right)}{t} = \i^n g\e^{-\i n \frac{\pi}{2}}2\pi\delta\left(t - n \frac{\pi}{\omega_c}\right) = 2\pi g \delta\left(t-\frac{n\pi}{\omega_c}\right)\ .
\end{align}

\subsection{Alternative calculation of $\gamma_0(t)$\label{app:gamma0}}
From Eq.~(\ref{eq:time-dep-couplings}), the coupling coefficient between the first chain mode $n = 0$ and the system is given by the convolution of the Fourier transform of the first polynomial $P_0 = 1$ and the Fourier transform of the rectangular function
\begin{align}
    \gamma_n(t) &= \sqrt{2\pi}g\Big(\mathcal{F}[P_n]*\mathcal{F}[\Pi_{\omega_c}]\Big)(t)\ . \label{eq:convolution}
\end{align}
It is well known that the Fourier transform of a rectangular function is a $\mathrm{sinc}$ function.
Precisely
\begin{align}
    \mathcal{F}[\Pi_{\omega_c}](t) & = \e^{-\i\frac{\omega_c t}{2}}\frac{\omega_c}{\sqrt{2\pi}}\sinc{\frac{\omega_c t}{2}}\ .
\end{align}
Hence, for the first mode we get
\begin{align}
    \gamma_0(t) &= \lim_{\omega_c\to\infty} \sqrt{2\pi}g\omega_c\sinc{\frac{\omega_c t}{2}}\e^{-\i\frac{\omega_c t}{2}} = (2\pi)^{\frac{3}{2}} g\delta(t)\label{eq:delta0}
\end{align}
as the Fourier transform of $P_0$ is proportional to the Dirac delta function (which is the neutral element of the convolution product).

\section{Numerical estimation of the maxima of spherical Bessel functions for finite $\omega_c$}
\label{app:numerical-threshold}
For finite $\omega_c$, one can estimate numerically the maxima of the spherical Bessel functions $\sqrt{\frac{\pi}{\omega_c t}}J_{n+\frac{1}{2}}\left(\frac{\omega_c t}{2}\right)$.
The locations $\omega_c t_n$ of these maxima are reported in Fig.~\ref{fig:maxima}, along with a linear fit $\omega_c t_n = 2.05123 n + 1.85029$ which gives a slope of $\simeq 2$ for the dependence of $\omega_c t_n$ on the chain mode label $n$ when $\omega_c$ is large but does not go towards infinity.
\begin{figure}[h]
    \centering
    \includegraphics[width=0.7\textwidth]{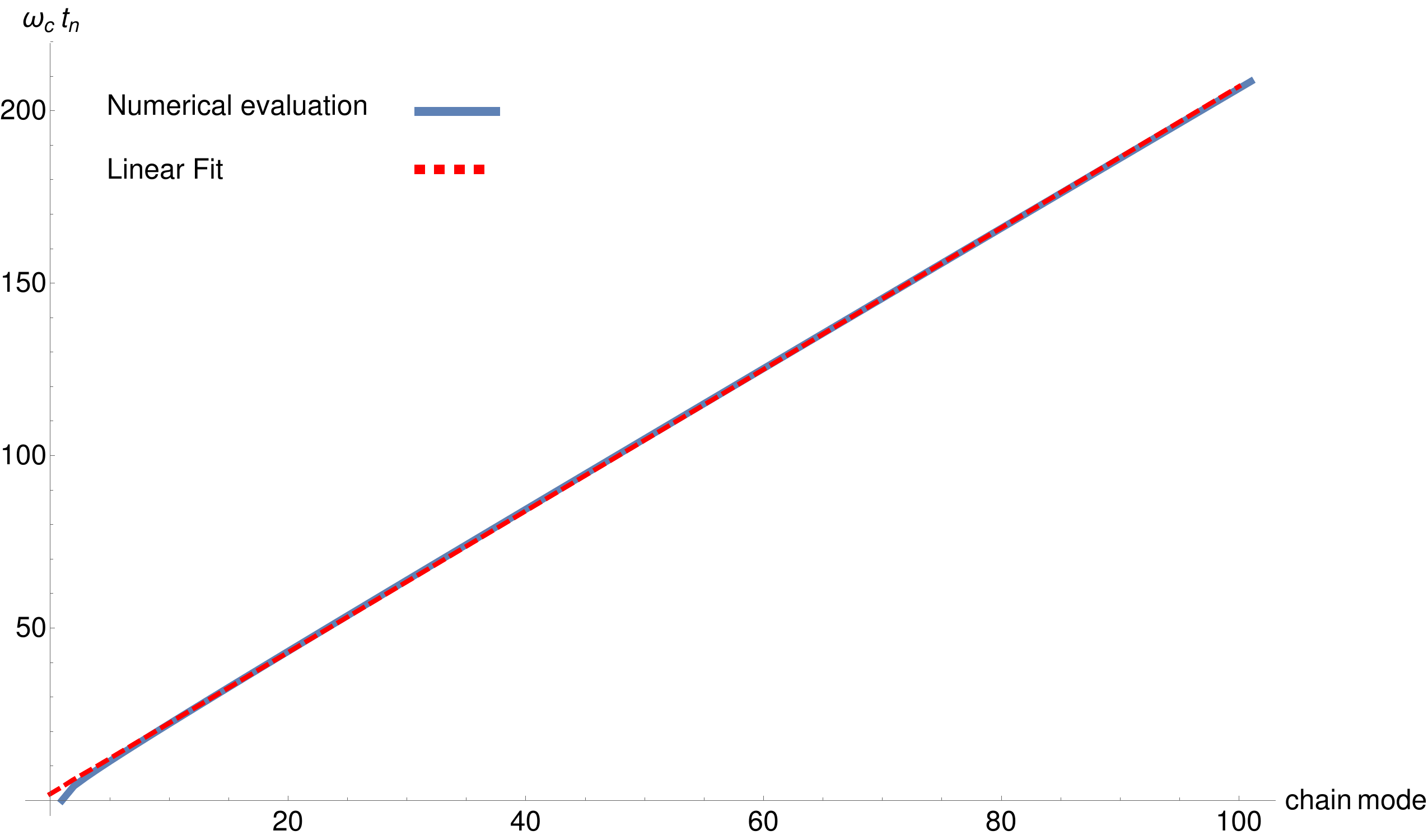}
    \caption{Numerical evaluation of the time $t_n$ such that $\gamma_n(t)$ is maximal against chain modes $n$, and a linear fit of this relation. The slope is of the order of $2$.}
    \label{fig:maxima}
\end{figure}

Figure~\ref{fig:coupling-strength}~(a) shows a heatmap of the time-dependent coupling strength $|\gamma_n(t)|$ computed via numerical integration for a finite $\omega_c$.
We can see that the further away along the chain a mode is (i.e.~the larger $n$ is), the later it will interact with the system.
Figure~\ref{fig:coupling-strength}~(b) shows the behaviour of a selection of time-dependent coupling strength (i.e. vertical cuts of the heatmap) and highlights their maxima.
\begin{figure}
    \centering
    \includegraphics[width=\textwidth]{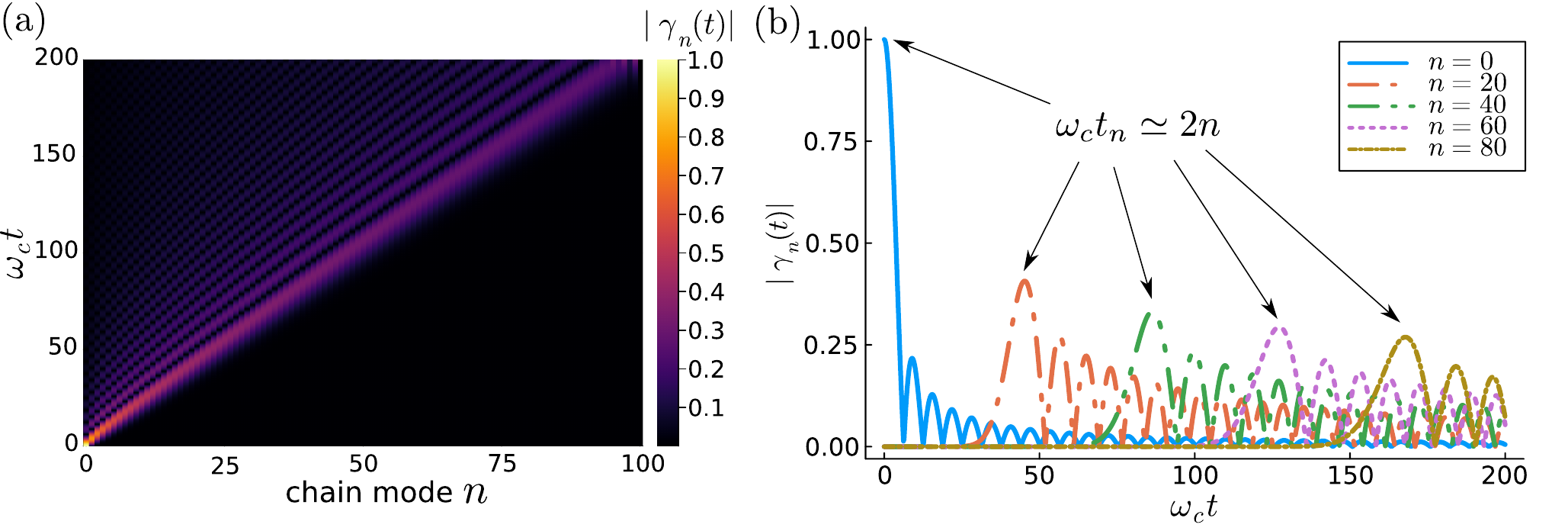}
    \caption{(a) Time-dependent coupling strength $|\gamma_n(t)|$ between the system and the chain modes $n$ for a flat SD. The oscillations in the coupling strength are induced by the finite bandwidth of the environment defined by the cut-off frequency $\omega_c$. (b) Time-dependence of the coupling strength for a subset of modes. In the limit of an infinite bandwidth ($\omega_c\to\infty$) the coupling strength reduces to a Dirac delta $|\gamma_n(t)| \propto \delta(t-t_n)$.}
    \label{fig:coupling-strength}
\end{figure}

\end{document}